\documentclass{article}

\usepackage{arxiv}

\usepackage[utf8]{inputenc} 
\usepackage[T1]{fontenc}    
\usepackage{hyperref}       
\usepackage{url}            
\usepackage{booktabs}       
\usepackage{amsfonts}       
\usepackage{nicefrac}       
\usepackage{microtype}      
\usepackage{lipsum}

\usepackage{amsfonts}
\usepackage{amsmath}
\usepackage{subcaption}
\usepackage{textcomp}
\usepackage{graphicx}
\usepackage{amssymb}
\usepackage[inkscapelatex=false]{svg}
\usepackage{mathtools}
\usepackage{newtxtext,newtxmath}
\usepackage{xcolor}
\usepackage{float}
\usepackage{algorithm}
\usepackage{algpseudocode}
\usepackage{stmaryrd}

\newcommand{\R}{{\mathbb{R}}}

\newcommand{\N}{{\mathbb{N}}}

\newcommand{\X}{{\mathbf{X}}}

\newcommand{\U}{{\mathbf{U}}}

\newtheorem{theorem}{Theorem}[section]

\newtheorem{assumption}{Assumption}

\newtheorem{definition}[theorem]{Definition}

\newtheorem{proof}[theorem]{Proof}
\newtheorem{remark}[theorem]{Remark}

\title{Formally Verified Neural Lyapunov Function for Incremental Input-to-State Stability of Unknown Systems
}

\author{
 Ahan Basu \\
  Centre for Cyber-Physical Systems\\
  IISc, Bengaluru, India\\
  \texttt{ahanbasu@iisc.ac.in} \\
   \And
 Bhabani Shankar Dey \\
  Centre for Cyber-Physical Systems\\
  IISc, Bengaluru, India\\
  \texttt{bhabanishan1@iisc.ac.in} \\
  \And
 Pushpak Jagtap \\
  Centre for Cyber-Physical Systems\\
  IISc, Bengaluru, India\\
  \texttt{pushpak@iisc.ac.in} \\
}

\begin{document}
\maketitle

\begin{abstract}
This work presents an approach to synthesize a Lyapunov-like function to ensure incrementally input-to-state stability ($\delta$-ISS) property for an unknown discrete-time system. To deal with challenges posed by unknown system dynamics, we parameterize the Lyapunov-like function as a neural network, which we train using the data samples collected from the unknown system along with appropriately designed loss functions. We propose a validity condition to test the obtained function and incorporate it into the training framework to ensure provable correctness at the end of the training. Finally, the usefulness of the proposed technique is proved using two case studies: a scalar non-linear dynamical system and a permanent magnet DC motor.
\end{abstract}


\section{Introduction}
Incremental stability mainly focuses on the convergence of the trajectories with respect to each other rather than the convergence to a nominal trajectory or equilibrium point. Numerous techniques have evolved to guarantee incremental stability over the years. Incremental stability has gained significant attention due to its applicability in the synchronization of cyclic feedback systems \cite{synch}, complex networks \cite{synchComplex} and interconnected oscillators \cite{synchOsci}, modelling of nonlinear analog circuits \cite{modelCirc}, and symbolic model construction for nonlinear control systems \cite{bisim1,bisim2,zamani2017towards,jagtap2020symbolic,jagtap2017quest}.

In the last few decades, researchers developed several tools to analyze incremental stability, which includes contraction analysis \cite{Contraction}, convergent dynamics \cite{Conv_dyn}, and incremental Lyapunov functions \cite{angeli2002lyapunov,DT-ISS_prop, zamani_inc}. These tools have also been extended to analyze the incremental stability of a wide class of systems, such as nonlinear systems \cite{angeli2002lyapunov}, stochastic systems \cite{biemond2018incremental}, hybrid dynamical systems \cite{biemond2018incremental},  and time-delayed systems \cite{chaillet2013razumikhin}. Apart from this, there have been several efforts in designing controllers to enforce incremental stability \cite{zamani2011backstepping,zamani2013backstepping,jagtap2017backstepping}. However, all these methods rely on the assumption of complete knowledge of the system dynamics. In contrast to this, recent work \cite{sundarsingh2024backstepping} develops an approach to estimate the unknown dynamics using Gaussian Process for a class of nonlinear systems and design a backstepping-like controller to ensure incremental stability. 

To deal with the unknown systems, learning-based approaches have attracted considerable attention. The use of neural networks to learn Lyapunov and barrier functions has become increasingly popular to ensure system stability and safety. This approach provides formal guarantees as demonstrated in previous studies \cite{Lyapunov_nn, formal_nn_Lyapunov, Neural_CBF, tayal2024learning}. Since the design of the Lyapunov function through a neural network-based approach is entirely data-driven, parameterizing this function as a neural network eliminates the need for an exact system model. This allows designers to bypass the constraints of predefining fixed templates for the Lyapunov function and instead directly synthesize it by leveraging the universal approximation property of neural networks. However, a significant challenge arises in ensuring formal guarantees because the network is trained on a discrete set of samples, representing only a subset of the actual continuous state space.

To the best of our knowledge, this is the first work to provide a formal incremental stability analysis for unknown systems. In particular, we present a novel training framework designed to synthesize verifiably correct incremental input-to-state stable Lyapunov functions, parameterized as neural networks, for unknown discrete-time systems. Notably, our approach eliminates the need for post-training formal verification. To approach this, we first establish a validity condition by formulating a Scenario Convex Problem (SCP), under the assumption that both the model dynamics and the neural network adhere to Lipschitz continuity. This ensures that the data-driven Lyapunov function obtained during training is verifiably correct. To further enhance the robustness of the training process, we impose a smaller Lipschitz constant on the neural network. This constraint reinforces the stability of the trained network, guaranteeing that the Lyapunov function remains valid across all points within the state space, thereby offering a formal correctness guarantee. We validate the effectiveness of our approach by applying it to two distinct case studies: a simple nonlinear scalar system and a permanent magnet DC motor. 
\section{Incremental Stability}

\subsection{Notations}
The symbols $\N$, $ \R$, $\R^+$, and $\R_0^+ $ denote the set of natural, real, positive real, and nonnegative real numbers, respectively. 
A vector space of real matrices with $ n $ rows and $ m $ columns is denoted by  $ \R^{n\times m} $. A column vector with $n$ rows is represented by $ \R^{n}$.
The Euclidean norm is represented using $|\cdot |$. 
Given a function $\varphi: \N \rightarrow \R^m$, its sup-norm (possibly infinite) is given by $\lVert \varphi \rVert = \{\sup|\varphi(k)| : k \in \N\}$.
For $a, b \in \mathbb{N}$ with $a \leq b$, the closed interval in $\mathbb{N}$ is denoted as $[a; b]$.  
A vector $x \in \mathbb{R}^{n}$ with entries $x_1, \ldots, x_n$ is represented as $[x_1, \ldots, x_n]^\top$, where $x_i \in \mathbb{R}$ denotes the $i$-th element of vector and $i \in [1;n]$.
A set with elements a,b,c is denoted by \{a,b,c\}. 
A diagonal matrix in $\R^{n\times n}$ with positive entries is denoted by $\mathcal{D}_{\geq 0}^n$.
Given a matrix $M\in\R^{n\times m}$, $M^\top$ represents transpose of matrix $M$. 
A continuous function $\alpha: \R_0^+ \rightarrow \R_0^+$ is said to be class $\mathcal{K}$ if $\alpha(s)>0 $ for all $s>0$, strictly increasing and $\alpha(0)=0$. It is class $\mathcal{K}_\infty$ if $\alpha(s)\rightarrow\infty$ for $s\rightarrow\infty$.
A continuous function $\beta: \R_0^+ \times \R_0^+ \rightarrow \R_0^+$ is said to be class $\mathcal{KL}$ if $\beta(s,t)$ is a class $\mathcal{K}$ function with respect to $s$ for all $t$ and for fixed $s>0$, $\beta(s,t)\rightarrow 0 $ if $t\rightarrow \infty$. A set $C\subseteq \mathbb{R}^n$ is defined to be forward invariant if for every $x\in C$, $\phi(t,x)$ will lie in the set $C$ for all $t\geq0$.

\subsection{Definitions}
In this work, we consider discrete-time systems as defined next.
\begin{definition}[Discrete-time Systems]\label{def:system}
     A discrete-time system (dt-CS) is represented by the tuple $\Xi = (\X, \U, f)$, where $\X \subseteq \R^n$ is the state-space of the system which is assumed to be forward-invariant, $\U \subseteq \R^m$ is the external input set of the system and $f: \X \times \U \rightarrow \R^n$ describes the state evolution via the following difference equation:
    \begin{equation}\label{eq:system}
        \mathsf{x}(k+1) = f(\mathsf{x}(k), \mathsf{u}(k)), \quad \forall k \in \N
    \end{equation}
    with $\mathsf{x}(k) \in \X$ and $\mathsf{u}(k) \in \U$ are the state and external input of the system at $k$-th instance, respectively.
\end{definition}

Let $\mathsf{x}_{x,\mathsf{u}}(k)$ be the state of the system \eqref{eq:system} at time instance $k$ starting from the initial state $x$ under the sequence of input $\mathsf{u}$. Next, we define the notion of incremental input-to-state ($\delta$-ISS) stability for the discrete-time system \eqref{eq:system}.

\begin{definition}[$\delta$-ISS \cite{Angeli}] \label{def:inc-stable_iss}
    The system in \eqref{eq:system} is said to be incrementally input-to-state stable ($\delta$-ISS) if there exists a class $\mathcal{KL}$ function $\beta$ and a class $\mathcal{K}_\infty$ function $\gamma$, such that for any $k \in \N$, for all $x, \hat x \in \X $ and any external input sequence $\mathsf{u},\hat{\mathsf{u}}$ the following holds:
    \begin{equation}\label{eq:gas-system}
        |\mathsf{x}_{x,\mathsf{u}}(k)-\mathsf{x}_{\hat x,\hat{\mathsf{u}}}(k)| \leq \beta(|x-\hat x|,k) + \gamma(\lVert \mathsf{u} - \hat{\mathsf{u}} \rVert).
    \end{equation}
\end{definition}
If $\mathsf{u} = \hat{\mathsf{u}}$, one can recover the notion of incremental global asymptotic stability as defined in \cite{Angeli}. Next, we define the $\delta$-ISS Lyapunov function.


\begin{definition}[$\delta$-ISS Lyapunov function \cite{DT-ISS_prop}]\label{def:ISS-Lf}
    A function $V:\R^n \times \R^n \rightarrow \R$ is said to be a $\delta$-ISS Lyapunov function for system $\Xi$ if there exist class $\mathcal{K}_\infty$ functions $\alpha_1, \alpha_2, \alpha_3$ and a class $\mathcal{K}$ function $\sigma$ such that:
    \begin{enumerate}\label{cond:ISS-Lf}
        \item[(i)] for all  $x,\hat{x}\in \X$, $\alpha_1(|x-\hat{x}|) \leq V(x,\hat{x}) \leq \alpha_2(|x-\hat{x}|),$
        \item[(ii)] for all $x,\hat{x}\in \X$ and for all $u, \hat{u} \in \U$, $V(f(x, u),f(\hat{x},\hat{u})) - V(x,\hat{x})\leq -\alpha_3(|x - \hat{x}|) + \sigma(|u-\hat{u}|).$
    \end{enumerate}
\end{definition}  

The following theorem describes $\delta$-ISS in terms of the existence of $\delta$-ISS Lyapunov function.

\begin{theorem}[\cite{DT-ISS}]\label{th:admit}
    The discrete-time control system $\Xi$ is $\delta$-ISS if it admits a $\delta$-ISS Lyapunov function.
\end{theorem}

\subsection{Problem Formulation}
In this work, we consider the following assumption. 
\begin{assumption}
    The function $f$ in \eqref{eq:system} is unknown, and we have access to the black box model of the system (i.e., given finitely many samples $\{\textbf{x}_1,\ldots, \textbf{x}_N\}$, $N\in\mathbb{N}, \textbf{x}_i \in \X$ and $\{\textbf{u}_1,\ldots, \textbf{u}_N\}, \textbf{u}_i \in \U$, we can have the consecutive state $\{f(\textbf{x}_1, \textbf{u}_1),\ldots,f(\textbf{x}_N, \textbf{u}_N)\}$). 
\end{assumption}
Given the aforementioned assumption, the primary objective of this paper is to find the $\delta$-ISS Lyapunov function to verify the $\delta$-ISS property for the systems with unknown dynamics. To solve this problem, we propose the construction of a neural network-based $\delta$-ISS Lyapunov function using finitely many data points. Further, we provide the formal guarantee on the learned neural $\delta$-ISS Lyapunov function.

\section{Neural $\delta$-ISS Lyapunov Function}

In this section, we want to find a $\delta$-ISS Lyapunov function that ensures the system is incrementally input-to-state stable according to Theorem \ref{th:admit}. To do so, we first reframe the conditions (i) and (ii) of Definition \ref{def:ISS-Lf} as a robust optimization program (ROP) defined next:

\begin{subequations} \label{eq:RCP}
\begin{align}
& \min_{\eta} \quad \eta  \notag \\
& \textrm{s.t.} \notag \\
& \forall x,\hat{x} \in \X: \notag \\
& \quad -V(x,\hat{x}) + \alpha_1(|x-\hat{x}|) \leq \eta, \label{eq:geq}\\
& \quad V(x,\hat{x}) - \alpha_2(|x-\hat{x}|) \leq \eta, \label{eq:leq}\\
& \forall x,\hat{x} \in \X, \forall u,\hat{u} \in \U: \notag \\
& \quad V(f(x, u),f(\hat{x},\hat{u}))\hspace{-.2em} - \hspace{-.2em}V(x,\hat{x}) \hspace{-.2em}+\hspace{-.2em} \alpha_3(|x\hspace{-.2em} -\hspace{-.2em} \hat{x}|) - \sigma(|u \hspace{-.2em}-\hspace{-.2em} \hat{u}|) \hspace{-.2em}\leq\hspace{-.2em} \eta. \label{eq:diff}
\end{align}
\end{subequations}
If the optimal solution of the ROP $\eta^* \leq 0$, then satisfying equation \eqref{eq:RCP} will imply satisfying conditions of Definition \ref{def:ISS-Lf}, and thus the corresponding Lyapunov function will be a valid $\delta$-ISS Lyapunov function for the unknown system, guaranteeing the system to be incrementally ISS. 

However, there are several challenges in solving the ROP. Since the function $f$ is unknown, incorporating condition \eqref{eq:diff} directly is nontrivial. So, finding the Lyapunov function is a complex task. Additionally, fixing the template of the $\delta$-ISS Lyapunov function may be restrictive; hence, finding suitable Lyapunov functions is not always possible. Therefore, to circumvent the issue, we parametrize the $\delta$-ISS Lyapunov function as a neural network, $V_{\theta,b}$, with $\theta$ (the weight parameters) and $b$ (the bias parameters).

For the training of the neural $\delta$-ISS Lyapunov function, one requires data from state space and input-space. Here, we leverage a sampling-based approach to get samples from the state-space $\X$ as well as the input space $\U$, and convert the ROP in \eqref{eq:RCP} into a scenario convex program (SCP) \cite{formal_nn_Lyapunov, Neural_CBF}. We collect samples $x_s$ and $u_p$ from $\X$ and $\U$, where $s \in [1;N], p \in [1;M]$. Consider ball $\X_s$ and $\U_p$ around each sample $x_s$ and $u_p$ with radius $\varepsilon_x$ and $\varepsilon_u$, such that, $\X \subseteq \bigcup_{s=1}^{N} \X_s$ and $\U \subseteq \bigcup_{p=1}^{M} \U_s$ with:
\begin{subequations}\label{eq:ball}
     \begin{align}
        |x - x_s| \leq \varepsilon_x , \forall x \in \X_s, \\
        |u - u_p| \leq \varepsilon_u , \forall u \in \U_p.
    \end{align}
\end{subequations}  
We consider $\varepsilon = \max(\varepsilon_x, \varepsilon_u)$.  The data sets obtained upon sampling the state-space $\X$ and input space $\U$ are denoted by:
\begin{subequations}\label{set:SCP}
    \begin{align}
    \mathcal{X} &= \{x_s | x_s \in \X, \forall s \in [1;N]\}, \\
    \mathcal{U} &= \{u_p | u_p \in \U, \forall p \in [1;M]\}.
    \end{align}
\end{subequations}

We consider the class $\mathcal{K}_\infty$ functions $\alpha_i, i \in \{1,2,3\}$ are of degree $\gamma_i$ with respect to $|x - \hat{x}|$ and and class $\mathcal{K}$ function $\sigma$ is of degree $\gamma_u$ with respect to $|u - \hat{u}|$ i.e., $\alpha_i(|x - \hat{x}|)=k_i|x - \hat{x}|^{\gamma_i}$, $i\in\{1, 2, 3\}$ and $\sigma(|u - \hat{u}|) = k_u|u - \hat{u}|^{\gamma_u}$, where $k_i, \gamma_i, i \in \{1,2,3,u\}$ are user-specific.

Utilizing the data sets and following the consideration above, we construct a scenario convex problem (SCP), defined below:
\begin{subequations} \label{eq:SCP}
\begin{align}
& \min_{\eta} \quad \eta  \notag \\
& \textrm{s.t.} \notag \\
& \forall x_q,x_r \in \mathcal{X}: \notag \\
& \quad -V_{\theta,b}(x_q,x_r) + k_1|x_q-x_r|^{\gamma_1} \leq \eta, \\
& \quad V_{\theta,b}(x_q,x_r) - k_2|x_q-x_r|^{\gamma_2} \leq \eta, \\
& \forall x_q,x_r \in \mathcal{X}, \forall u_q, u_r \in \mathcal{U}: \notag \\
& \quad V_{\theta,b}(f(x_q, u_q),f(x_r, u_r)) - V_{\theta,b}(x_q,x_r) + k_3 |x_q - x_r|^{\gamma_3} - k_u|u_q - u_r|^{\gamma_u}\leq \eta. 
\end{align}
\end{subequations}
Since we have finite data samples, the number of equations in the SCP is finite, hence the solution of the SCP is tractable. Let the optimal solution of the SCP be $\eta_S^*$. Now to prove the solution of SCP is also a feasible solution to the proposed RCP, we make the following assumptions on Lipschitz continuity:

\begin{assumption}\label{assum:Lipschitz_fun}
    The function $f$ in \eqref{eq:system} is Lipschitz continuous with respect to $x$ and $u$ over the state space $\X$ and input space $\U$ with Lipschitz constant $\mathcal{L}_x$ and $\mathcal{L}_u$. One can estimate these constants following the similar procedure in \cite[Algorithm 2]{FV_DD}. 
\end{assumption}

\begin{assumption}\label{assum:Lipschitz_net}
    Moreover, we assume the candidate neural Lyapunov function to be Lipschitz continuous with Lipschitz constant $\mathcal{L}_L$ with respect to $(x,\hat{x})$ over the set $\X$. In the next section, we will explain how $\mathcal{L}_L$ is to be ensured in the training procedure.
\end{assumption}

\begin{remark}
    The class $\mathcal{K}_\infty$ functions and class $\mathcal{K}$ function are Lipschitz continuous with Lipschitz constants $\mathsf{L}_1, \mathsf{L}_2, \mathsf{L}_3$ and $\mathsf{L}_u$ respectively with respect to $|x-\hat{x}|$ and $|u-\hat{u}|$. One can estimate the values using the values of $k_i$ and $\gamma_i, i \in \{1,2,3,u\}$.
\end{remark}

Under Assumption \ref{assum:Lipschitz_fun} and \ref{assum:Lipschitz_net}, Theorem \ref{th:constr} outlines the connection of the solution of SCP \eqref{eq:SCP} to that of RCP \eqref{eq:RCP}, providing formal guarantee to the obtained Lyapunov function satisfying the incremental stability conditions.
\begin{theorem}\label{th:constr}
    Consider a dt-CS represented as $\Xi$. Let $V_{\theta,b}$ be the neural $\delta$-ISS Lyapunov function. The optimal value of the SCP in \eqref{eq:SCP}, $\eta_S^*$, is obtained using the sampled data. Then under Assumption \ref{assum:Lipschitz_fun} and \ref{assum:Lipschitz_net}, if
    \begin{equation}\label{eq:cond}
        \eta_S^* + \mathcal{L}\varepsilon \leq 0, 
    \end{equation}
    where $\mathcal{L} = \max\{\sqrt{2}\mathcal{L}_L + 2\mathsf{L}_1, \sqrt{2}\mathcal{L}_L + 2\mathsf{L}_2, \sqrt{2}\mathcal{L}_L(\mathcal{L}_x + \mathcal{L}_u + 1) + 2(\mathsf{L}_3 + \mathsf{L}_u)\}$, then the $\delta$-ISS Lyapunov function obtained by solving the SCP ensures the system to be incrementally input-to-state stable.
\end{theorem}
\begin{proof}
    Under the condition \eqref{eq:cond}, we demonstrate that the obtained Lyapunov function from SCP satisfies the conditions of Definition \ref{def:ISS-Lf}. The optimal $\eta_S^*$, obtained through solving the \eqref{eq:SCP}, guarantees for any $x_q,x_r \in \mathcal{X}, u_q,u_r \in \mathcal{U}$, we have: 
    \begin{align*}
    & -V_{\theta,b}(x_q,x_r) + k_1|x_q-x_r|^{\gamma_1} \leq \eta_S^*,\\
    & V_{\theta,b}(x_q,x_r) - k_2|x_q-x_r|^{\gamma_2} \leq \eta_S^*, \\ 
    & V_{\theta,b}(f(x_q, u_q),f(x_r, u_r)) - V_{\theta,b}(x_q,x_r) + k_3 |x_q - x_r|^{\gamma_3} - k_u|u_q - u_r|^{\gamma_u}\leq \eta_S^*.
    \end{align*}
Now from \eqref{eq:ball}, we infer that $\forall x \in \X, \text{there exists}$ $ x_r $ s.t. $|x-x_r| \leq \varepsilon$ as well as $\forall u \in \U, \text{there exists}$ $ u_r $ s.t. $|u-u_r| \leq \varepsilon$. Hence, $\forall x,\hat{x} \in \X, \forall u, \hat{u} \in \U$:
    \begin{align*}
        \text{(a)} & -V_{\theta,b}(x,\hat{x}) + k_1|x-\hat{x}|^{\gamma_1} \\
        & = \big(-V_{\theta,b}(x,\hat{x}) + V_{\theta,b}(x_q,x_r)\big) + \big( - V_{\theta,b}(x_q,x_r) + k_1|x_q-x_r|^{\gamma_1}\big) + \big(- k_1|x_q-x_r|^{\gamma_1} + k_1|x-\hat{x}|^{\gamma_1}\big) \\
        & \leq \mathcal{L}_L |(x,\hat{x}) - (x_q,x_r)| + \eta_S^* + 2\mathsf{L}_1\varepsilon\\
        & \leq \sqrt{2}\mathcal{L}_L\varepsilon + \eta_S^* + 2\mathsf{L}_1\varepsilon \leq \mathcal{L}\varepsilon + \eta_S^* \leq 0. \\
        \text{(b)} & V_{\theta,b}(x,\hat{x}) - k_2|x-\hat{x}|^{\gamma_2} \\
        & = \big(V_{\theta,b}(x,\hat{x}) - V_{\theta,b}(x_q,x_r)\big) + \big(  V_{\theta,b}(x_q,x_r) - k_2|x_q-x_r|^{\gamma_2}\big) + \big(k_2|x_q-x_r|^{\gamma_2} - k_2|x-\hat{x}|^{\gamma_2}\big) \\
        & \leq \mathcal{L}_L |(x,\hat{x}) - (x_q,x_r)| + \eta_S^* + 2\mathsf{L}_2\varepsilon\\
        & \leq \sqrt{2}\mathcal{L}_L\varepsilon + \eta_S^* + 2\mathsf{L}_2\varepsilon \leq \mathcal{L}\varepsilon + \eta_S^* \leq 0. \\
        \text{(c)}& V_{\theta,b}(f(x, u),f(\hat{x},\hat{u})) - V_{\theta,b}(x,\hat{x}) + k_3|x - \hat{x}|^{\gamma_3} - k_u|u - \hat{u}|^{\gamma_u}\\
         & = \big(V_{\theta,b}(f(x, u),f(\hat{x},\hat{u})) - V_{\theta,b}(f(x_q, u_q),f(x_r,u_r))\big) + \big(V_{\theta,b}(f(x_q, u_q),f(x_r,u_r)) - V_{\theta,b}(x_q,x_r) \\ 
         & \quad + k_3 |x_q - x_r|^{\gamma_3}- k_u|u_q - u_r|^{\gamma_u}\big) + \big(V_{\theta,b}(x_q,x_r) - V_{\theta,b}(x,\hat{x})\big) + \big(k_3(|x - \hat{x}|^{\gamma_3} - |x_q - x_r|^{\gamma_3})\big)\\
         & \quad + \big(k_u(|u_q - u_r|^{\gamma_u} - |u - \hat{u}|^{\gamma_u})\big)\\
         & \leq \mathcal{L}_L|(f(x, u),f(\hat{x},\hat{u})) - (f(x_q, u_q),f(x_r,u_r))| + \eta_S^* + \sqrt{2}\mathcal{L}_L\varepsilon + 2(\mathsf{L}_3+\mathsf{L}_u)\varepsilon \\ 
         & \leq \sqrt{2}\mathcal{L}_L(\mathcal{L}_x + \mathcal{L}_u) \varepsilon + \eta_S^* + \sqrt{2}\mathcal{L}_L\varepsilon + 2(\mathsf{L}_3+\mathsf{L}_u)\varepsilon \\
         & \leq (\sqrt{2}\mathcal{L}_L(\mathcal{L}_x + \mathcal{L}_u + 1) + 2(\mathsf{L}_3+\mathsf{L}_u)) \varepsilon + \eta_S^* \\
         & \leq \mathcal{L}\varepsilon + \eta_S^* \leq 0.
     \end{align*}
    Therefore, if the condition \eqref{eq:cond} is satisfied, the neural Lyapunov function will imply incrementally input-to-state stability of the system. This completes the proof.
\end{proof}


\section{Training Provably Correct Neural $\delta$-ISS Lyapunov Function }
In this section, we utilize the derived validity condition \eqref{eq:cond} from Section III and propose a training framework to synthesize provably correct $\delta$-ISS Lyapunov function parametrized as neural networks. We train the neural network to achieve a formal guarantee of its validity by constructing suitable loss functions that incorporate the satisfaction of conditions \eqref{eq:geq}-\eqref{eq:diff} and \eqref{eq:cond}.
\subsection{Neural Network Structure}
Given a dt-CS $\Xi$, we define the neural $\delta$-ISS Lyapunov function as $V_{\theta,b}$, where the neural network is parametrized by the weights $\theta$ and bias $b$. It consists of an input layer with $2n$ ($2$ times the system dimension) neurons and an output layer with one neuron due to the scalar value of the Lyapunov function. The number of hidden layers is denoted by $l_f$, and each hidden layer has $h_f^i$, where $i \in [1;l_f]$ neurons, where both $l_f$ and $h_f^i$ are arbitrarily chosen.

The activation function of all the layers except the output layer is chosen to be any slope-restricted function $\phi(\cdot)$ (for example, ReLU, Sigmoid, Tanh etc.). Hence the resulting neural network function is obtained by applying the activation function recursively and is denoted by:
\[
\begin{cases}
x^0 = [x^\top,\hat{x}^\top]^\top , x,\hat{x} \in \R^n, \\
x^{i+1} = \phi(\theta^ix^i + b^i) \hspace{0.2 em} \text{for} \hspace{0.2 em} i \in [0;l_f-1], \\
V_{\theta,b}(x,\hat{x}) = \theta^{l_f}x^{l_f} + b^{l_f}. 
\end{cases}
\]
Note that $\phi(\cdot)$ is applied element-wise. 
\subsection{Training with Formal Guarantees}
This section discusses the procedure of training the neural $\delta$-ISS Lyapunov function $V_{\theta,b}$ for the dt-CS $\Xi$ ensuring the satisfaction of the validity condition \eqref{eq:cond} such that the trained Lyapunov function satisfies the conditions of Definition \ref{def:ISS-Lf} over the state space $\X$ and input set $\U$. We need to assume the Lipschitz continuity of the $\delta$-ISS candidate Lyapunov function as per Assumption \ref{assum:Lipschitz_net}. Since the neural network consists of slope-restricted activation layers, the Lipschitzness assumption is valid \cite{Neural_CBF}. Based on the assumption, let us consider $V_{\theta,b}$ is Lipschitz continuous with a bound $\mathcal{L}_L$, and hence:
\begin{align}
    \forall w,x,y,z \in \R^n:
    |V_{\theta,b}(w,x) - V_{\theta,b}(y,z)| \leq \mathcal{L}_L|(w,x) - (y,z)|.
\end{align}
Then, one can guarantee the Lipschitz constant of $V_{\theta,b}$ is bounded by $\mathcal{L}_L$ if the following inequality \cite{pauli2022a} holds:
\begin{equation}\label{eq:mat_ineq}
    \underbrace{\begin{bmatrix}
    \mathcal{L}_L^2I_{2n} & -\theta^{0^T}\Lambda_1 & 0 & \ldots & 0 \\
    -\Lambda_1\theta^0 & 2\Lambda_1 & \ddots & \ddots & \vdots \\
    0 & \ddots & \ddots & -\theta^{{l_f-1}^T}\Lambda_{l_f} & 0 \\
    \vdots & \ddots & -\Lambda_{l_f}\theta^{l_f-1} & 2\Lambda_{l_f} & -\theta^{{l_f}^T} \\
    0 & \ldots & 0 & -\theta^{l_f} & I_{h_l^{l_f}}
    \end{bmatrix}}_{:=\mathcal{P}(\theta,\Lambda)} \geq 0.
\end{equation}

Here, $\Lambda = (\Lambda_1,\ldots,\Lambda_{l_f}), \Lambda_i \in \mathcal{D}_{\geq 0}^{h_b^i}, i \in [1;l_f]$. The matrix inequality $\mathcal{P}(\theta,\Lambda)\geq 0$ can be obtained using a similar method as in \cite{pauli2022b}.

Now we describe the formulation of suitable loss functions for the training of $V_{\theta,b}$ such that minimization of the loss function leads to the satisfaction of the conditions in SCP \eqref{eq:SCP} over the training data set \eqref{set:SCP} and satisfying the validity condition \eqref{eq:cond}. We consider \eqref{eq:geq}-\eqref{eq:diff} as sub-loss functions to construct the actual loss function, which can be renamed as the Lyapunov risk function. The sub-loss functions are:
\begin{subequations}\label{eq:loss_LR}
    \begin{align}
        L_0(\psi,\eta) &= \sum_{x,\hat{x} \in \mathcal{X}}\max\big(0,(-V_{\theta,b}(x,\hat{x}) + k_1|x- \hat{x}|^{\gamma_1} - \eta)\big), \\
        L_1(\psi,\eta) &= \sum_{x,\hat{x} \in \mathcal{X}}\max\big(0,(V_{\theta,b}(x,\hat{x}) - k_2|x- \hat{x}|^{\gamma_2} - \eta)\big), \\
        L_2(\psi,\eta) &= \sum_{x,\hat{x} \in \mathcal{X}, u, \hat{u} \in \mathcal{U}}\max\big(0,(V_{\theta,b}(f(x, u),f(\hat{x},\hat{u})) - V_{\theta,b}(x,\hat{x}) + k_3|x - \hat{x}|^{\gamma_3} - k_u|u - \hat{u}|^{\gamma_u} - \eta)\big),
    \end{align}
\end{subequations}
where, $\psi = [\theta,b]$ and $\eta$ are trainable parameters. Finally, the Lyapunov risk is a weighted sum of the sub-loss functions and is denoted by:
\begin{equation}
    L(\psi, \eta) = c_0L_0(\psi, \eta) + c_1L_1(\psi, \eta) + c_2L_2(\psi, \eta)
\end{equation}
where $c_0,c_1, c_2 \in \R^+$ are the weight of the sub-loss functions $L_0(\psi, \eta),L_1(\psi, \eta),L_2(\psi, \eta)$ respectively. Now for the satisfaction of conditions \eqref{eq:cond} and \eqref{eq:mat_ineq}, we consider two more loss functions denoted by:
\begin{align}
    L_{\mathcal{P}}(\psi,\Lambda) &= -c_l\log\det(\mathcal{P}(\theta,\Lambda)), \label{eq:loss_ineq}\\
    L_v(\eta) &= \max \big(0,(\mathcal{L}\varepsilon + \eta)\big), \label{eq:loss_th}
\end{align}
where $\psi,\Lambda$ are trainable parameters and $\mathcal{L}_L$ that appears in $\mathcal{P}(\theta,\Lambda)$ is used to compute $\mathcal{L}$ as described in Theorem \ref{th:constr}. The Lipschitz bound along with the $\varepsilon$ are hyper-parameters that are chosen apriori. Addtionally, $-c_l \in \R^+$. 

Now under the trained neural network corresponding to the Lyapunov-like function for the dt-CS $\Xi$, the following theorem provides a formal guarantee for the system to be incrementally input-to-state stable.
\begin{theorem}\label{th:guarantee}
    Consider a dt-CS in $\Xi$ with state-space $\X$ and input space $\U$. Let, $V_{\theta,b}$ denotes the trained neural network representing $\delta$-ISS Lyapunov function such that $L(x,\hat{x}), L_v(\eta) = 0$ and $L_{\mathcal{P}}(\psi,\Lambda) \leq 0$ over the training data set $\mathcal{X}$ and $\mathcal{U}$ as mentioned in \eqref{set:SCP}. Then, the system $\Xi$ is guaranteed to be incrementally ISS within the state space under the influence of any elements in the input space.
\end{theorem}

The training process of the neural Lyapunov function is described below:
\begin{enumerate}
    \item \textbf{Require:} Set of sampled data point $\mathcal{X}$ and $\mathcal{U}$, and the black box model of the system representing $f(x,u)$.
    \item Fix all the hyper parameters $c = [c_0, c_1, c_2, c_l], k = [k_1, k_2, k_3, k_u], \gamma = [\gamma_1, \gamma_2, \gamma_3, \gamma_u], \mathcal{L}_L$ apriori. Fix the number of epochs $n_{ep}$ apriori as well.
    \item Estimate the Lipschitz constant $\mathcal{L}_x$ and $\mathcal{L}_u$ using Reverse Weibull distribution \cite{Lipschitz}. Get the Lipschitz constant $\mathcal{L}$ as defined in Theorem \ref{th:constr}.
    \item Randomly generate several batches $\X_k, k = [1;n_b]$ from the dataset. Set the number of batches $n_b$ apriori.
    \item In every epoch, pass each batch to the network and get the output $V_{\theta, b}(x,\hat{x})$ for that batch. Pass the same batch of data into the black box model of $f$ and obtain the batch of next states $\X_k^+, k = [1;n_b]$. Pass the next state batch to the network again to get $V_{\theta, b}(f(x, u),f(\hat{x}, \hat{u}))$. Calculate the loss for each batch using the loss functions mentioned in \eqref{eq:loss_LR}, \eqref{eq:loss_ineq} and \eqref{eq:loss_th}. The cumulative sum of batch losses will give the epoch loss.
    \item Utilizing Adam or Stochastic Gradient Descent (SGD) optimization techniques with a specified learning rate \cite{ruder2016}, reduce the loss function and update the trainable parameters $\phi,\eta$. The learning rate can be different for different parameters.
    \item Repeat steps 6 and 7 until the loss functions converge according to Theorem \ref{th:guarantee}. After successful convergence of the training, the neural network will act as $\delta$-ISS-Lyapunov function $V_{\theta,b}$, and the system is considered to be incrementally input-to-state stable under the influence of the bounded input state $\U$. 
\end{enumerate}
\begin{remark}
    If the algorithm does not converge successfully, one can not comment on the incremental input-to-state stability of the dt-CS $\Xi$ with the specified hyperparameters $c,k,\gamma, \mathcal{L}_L$. 
\end{remark}

\begin{remark}
    Also, the initial feasibility of condition \eqref{eq:mat_ineq} is required to satisfy the criterion of loss $L_\mathcal{P}$ in \eqref{eq:loss_ineq} according to Theorem \ref{th:guarantee}. Choosing small initial weights and bias of the neurons can ensure this condition \cite{pauli2022a}.
\end{remark}


\section{Case Study}

The proposed verification procedure of an Incremental Stable system using a neural network-based $\delta$-ISS Lyapunov function is shown through the following two case studies: $(i)$ a nonlinear numerical example and (ii) a permanent magnet DC motor.

All the computations were performed using PyTorch in Python 3.8.10 on a machine with a Linux Ubuntu operating system with Intel i7-7700 CPU and 32GB RAM.

\subsection{Numerical Example}
We consider a simple non-linear one-dimensional system, whose discrete-time dynamics is given by:
\begin{align}
    \mathsf{x}(k+1) = \mathsf{x}(k) + \tau(a \sqrt{\mathsf{x}(k)} + \mathsf{u}(k)),
    \label{eq:case1}
\end{align}
where $\mathsf{x}(k)$ denotes the state of the system at $k$-th instant. The constant $a = -1$ stands for the decay rate constant of the system. $\tau = 0.01$ is the sampling time. We consider the state space of the system to be $\X = [0, 0.5]$. 
Also, we consider the model to be unknown. However, we estimate the Lipschitz constant $\mathcal{L}_x = 1, \mathcal{L}_u = 0.01$ using the results in \cite{Lipschitz}. 

The goal is to synthesize a valid Lyapunov-like function $V_{\theta,b}$ that will ensure the system to be incrementally input-to-state stable under the influence of a bounded input space $\U = [0, 0.5]$. We first fix the training hyper-parameters as $\varepsilon = 0.000177, \mathcal{L}_L = 1.5, k_1 = 0.00001, k_2 = 1, k_3 = k_u = 0.0001$. So, the Lipschitz constant $\mathcal{L} = 4.264$ using Theorem \ref{th:constr}. We fix the structure of $V_{\theta,b}$ as $l_f = 1, h_f^1 = 20$. 

Now we consider the training data obtained from \eqref{set:SCP} and perform training to minimize the loss functions $L, L_\mathcal{P} $ and $ L_v$. The training algorithm converges to obtain the $\delta$-ISS Lyapunov function $V_{\theta,b}$ along with $\eta = -0.0008$. Hence, $\eta+\mathcal{L}\varepsilon = -0.0008 + 4.264\times0.000176 = -0.00005$, thus by utilizing Theorem \ref{th:guarantee}, we can guarantee the obtained $\delta$-ISS Lyapunov function $V_{\theta,b}$ is valid and the system is guaranteed to be incrementally input-to-state stable.

The successful runs of the algorithm have an average convergence time of 10 minutes, and training data generation takes an additional time of 1 second.

One can see from Figure \ref{fig:sim1}, that the trajectories starting from different initial conditions are convergent towards each other and the difference between the trajectories at steady state is caused due to the difference in inputs. Hence, the system is said to be incrementally ISS.
\begin{figure}[h]
    \centering
    \includegraphics[width=\linewidth]{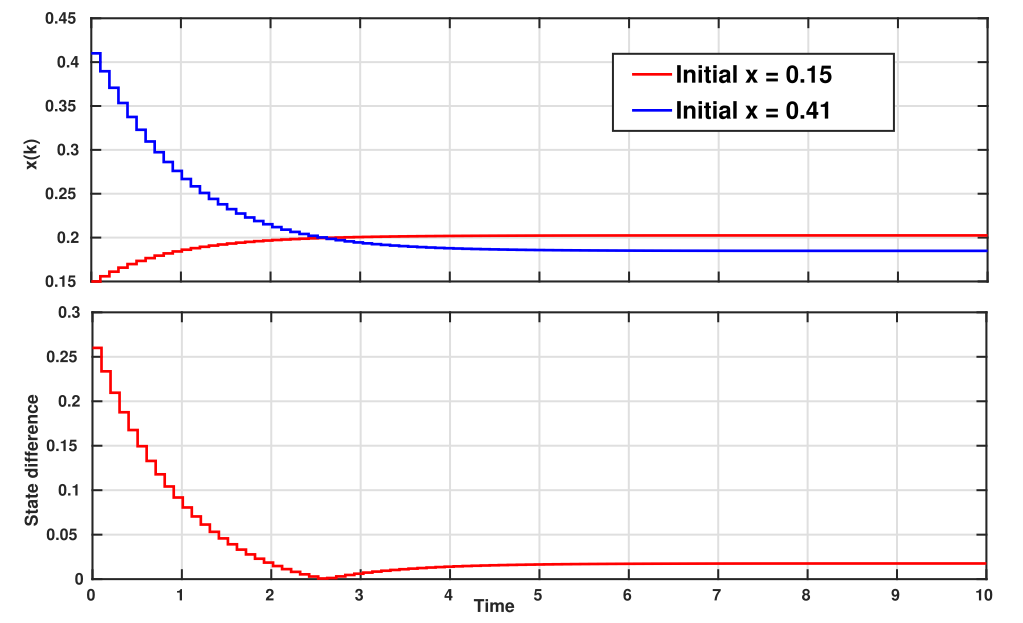}
    \caption{Top: State response of a simple nonlinear system as in \eqref{eq:case1}, where the blue curve is influenced under input $u$ = 0.43, and the red curve is influenced under input $u$ = 0.45. Bottom: Difference between trajectories subjected to different initial conditions and inputs.}
    \label{fig:sim1}
\end{figure}

\subsection{Permanent Magnet DC Motor}
We consider a permanent magnet DC motor adopted from \cite{DC_motor}, whose discrete-time dynamics is given by:
\begin{align*}
    \mathsf{x}_1(k+1) &= \mathsf{x}_1(k) + \tau \Big(-\frac{R_a}{L_a}\mathsf{x}_1(k) - \frac{k_b}{L_a}\mathsf{x}_2(k) + \frac{V_{in}}{L_a}\Big) \\
    \mathsf{x}_2(k+1) &= \mathsf{x}_2(k) + \tau \Big(\frac{k_b}{J}\mathsf{x}_1(k) - \frac{B}{J}\mathsf{x}_2(k)\Big),
\end{align*}
where $\mathsf{x}_1(k),\mathsf{x}_2(k)$ denote the armature current and the rotational speed of the shaft of the motor, respectively. The constants $R_a = 1, L_a = 0.01, J = 0.01, b = 1, k_b = 0.01$ stand for the armature resistance, inductance, moment of inertia of the rotor, motor torque, and back electromotive force of the motor respectively. The sampling time $\tau$ is considered to be $0.001$. The input voltage $V_{in}$ is fixed at $0.17$ volt. We consider the state space of the system to be $\X = [0, 0.2]\times[0, 0.2]$. Also, we consider the model to be unknown. However, we estimate the Lipschitz constants $\mathcal{L}_x = 0.048$ using the results in \cite{Lipschitz}. 

Here, we first fix the training hyper-parameters as $\varepsilon = 0.004, \mathcal{L}_L = 1, k_1 = 0.00001, k_2 = 0.04, k_3 = k_u = 0.0001$. So, the Lipschitz constant $\mathcal{L} = 1.4962$. We fix the structure of $V_{\theta,b}$ as $l_f = 1, h_f^1 = 20$. The training algorithm converges to obtain the $\delta$-ISS Lyapunov function $V_{\theta,b}$ along with $\eta = -0.01$. Exploiting theorem \ref{th:guarantee}, we can guarantee the obtained $\delta$-ISS Lyapunov function $V_{\theta,b}$ is valid and the system is guaranteed to be incrementally input-to-state stable.

The successful runs of the algorithm have an average convergence time of $15$ minutes, and training data generation takes an additional time of $2$ seconds.

One can see from Figure \ref{fig:sim2} and \ref{fig:sim3} that the trajectories starting from different initial conditions are convergent towards each other, and the difference between the trajectories in steady state is because of the difference in input signals. Hence, the system is incrementally ISS.



\begin{figure}[h]
    \centering
    \includegraphics[width=\linewidth]{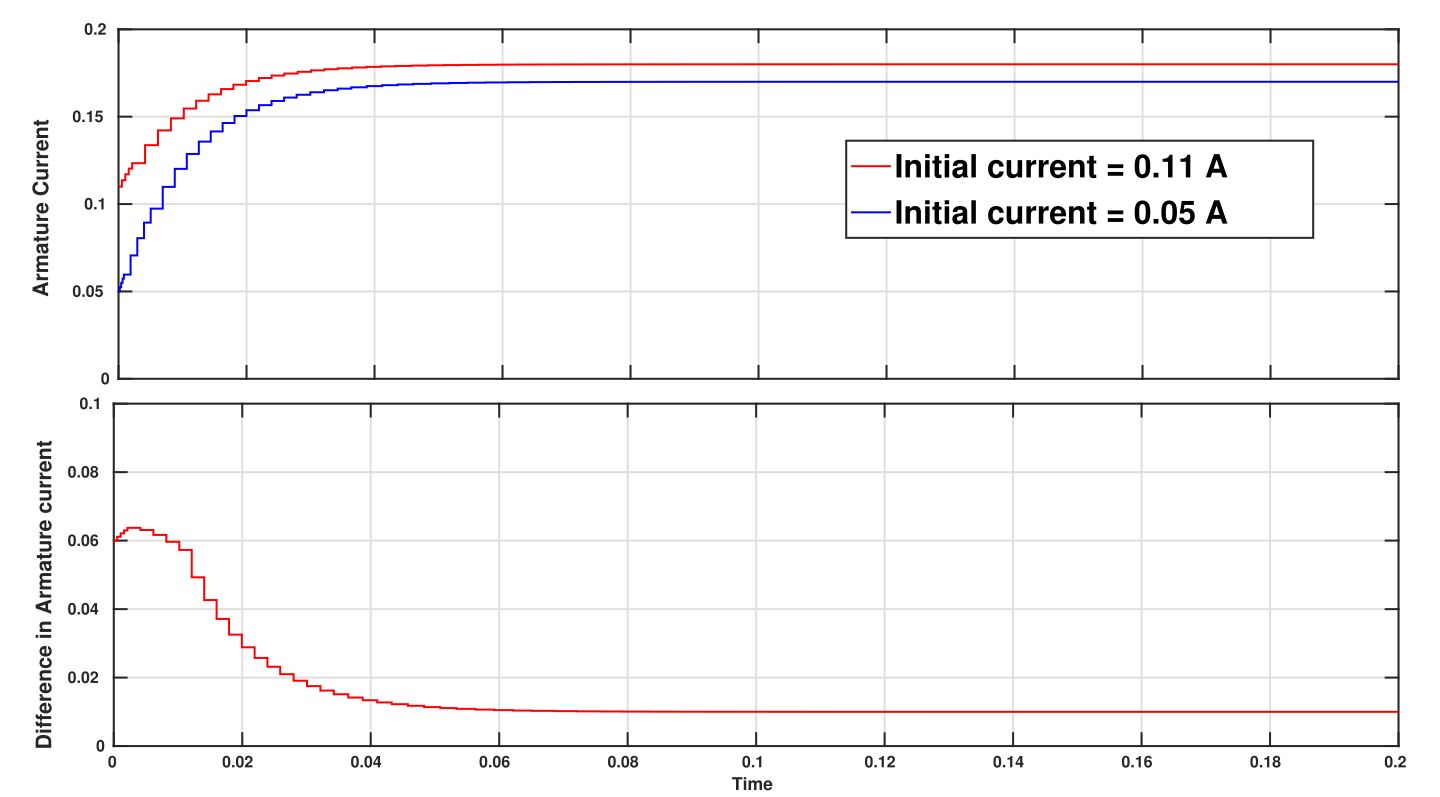}
    \caption{Top: Armature Current for DC motor, where the blue curve is influenced under input Vin = 0.17, and the red curve is influenced under input Vin = 0.18. Bottom: The difference in armature currents subjected to different initial conditions and input voltages.}
    \label{fig:sim2}
\end{figure}

\begin{figure}[h]
    \centering
    \includegraphics[width=\linewidth]{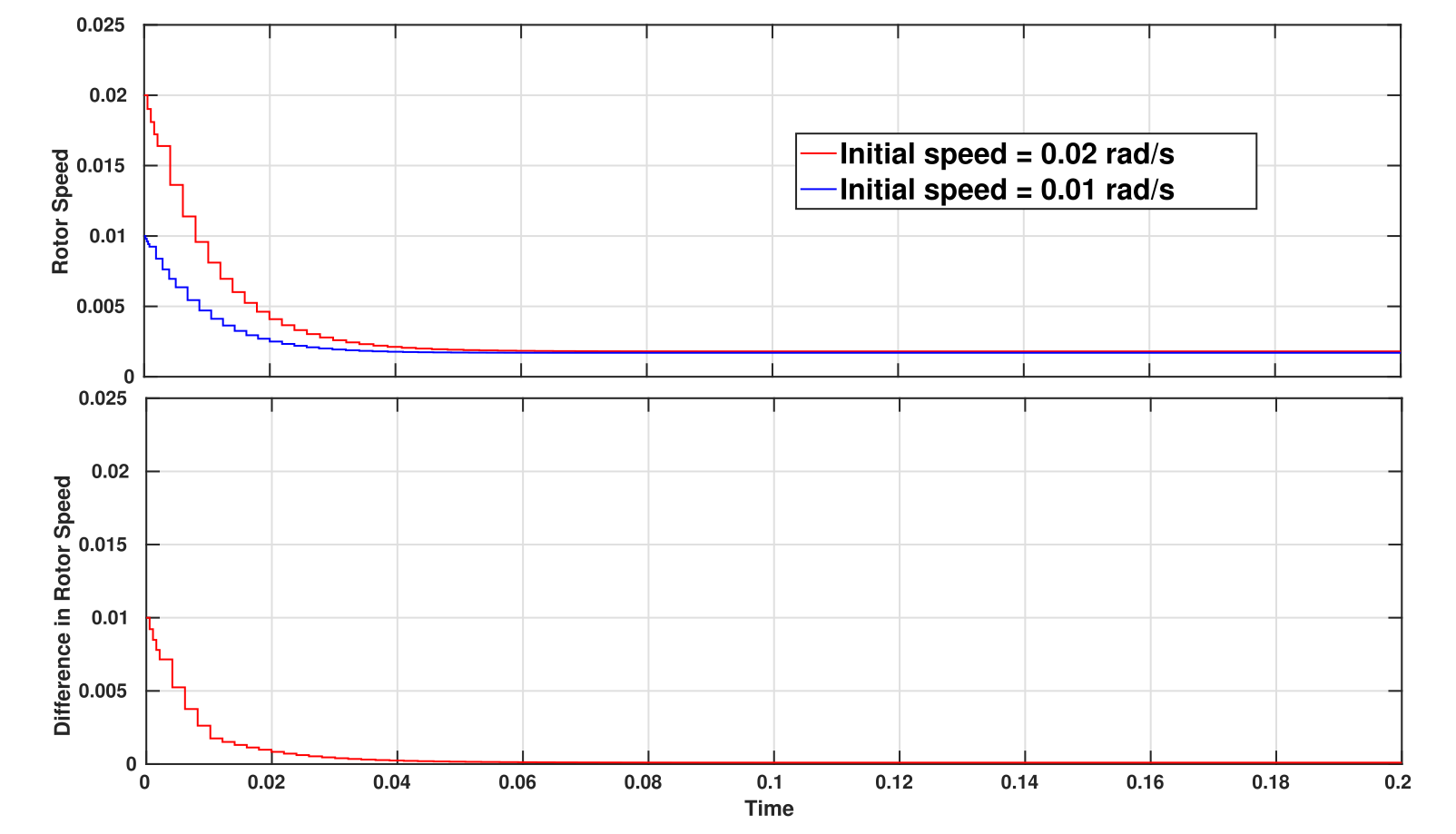}
    \caption{Top: The rotor speed for DC motor, where the blue curve is influenced under input Vin = 0.17, and the red curve is influenced under input Vin = 0.18. Bottom: The difference in rotor speeds subjected to different initial conditions and input voltages.}
    \label{fig:sim3}
\end{figure}
\section{Conclusion and Future work}
In this paper, we have successfully demonstrated how to synthesize a formally verified neural network that acts as an incremental input-to-state stable Lyapunov function for an unknown dynamical system. Here, we formulate the constraints of $\delta$-ISS Lyapunov function into an RCP, collect data, and formulate SCP. Next, we provide a validity condition that guarantees the successful solution of the SCP is valid for RCP as well. Hence the obtained neural network becomes a valid $\delta$-ISS Lyapunov function. Then, we utilize the validity condition and propose the training framework to synthesize provably correct Lyapunov function and formally guarantee its validity by adjusting suitable loss functions. The work can be extended into constructing $\delta$-ISS Lyapunov functions for continuous-time counterparts. One possible future direction of the work is to design a neural network-based controller to ensure a system is incrementally input-to-state stable. 

\bibliographystyle{ieeetr} 
\bibliography{sources} 
\end{document}